\documentclass[journal]{IEEEtran}

\onecolumn

\usepackage[boxed,vlined,ruled]{algorithm2e}
\usepackage[utf8]{inputenc} 
\usepackage[T1]{fontenc}
\usepackage{url,lipsum}
\usepackage{ifthen}
\usepackage{hyperref}
\hypersetup{
    colorlinks=true,
    linkcolor=black,
    filecolor=black,      
    urlcolor=black,
    citecolor = blue,
}
\usepackage[inline, shortlabels]{enumitem}
\setlist{itemjoin ={,\enspace},itemjoin* = {, and\enspace}}

\usepackage{algorithmic,nicefrac}
\usepackage{array,graphicx,pifont}
\usepackage[misc,clock,geometry]{ifsym}
\usepackage{amsfonts,cite,euscript}
\usepackage[cmex10]{amsmath}
\usepackage[amsmath]{empheq}
\usepackage{cases,balance,cite, lineno}
\usepackage{booktabs,mathrsfs,bbm,amssymb,soul,lipsum}
\usepackage{enumitem,amsthm}
\usepackage{dsfont}
\usepackage[scaled=.85]{beramono}
\usepackage[permil]{overpic}
\usepackage{graphicx}
\usepackage{url}
\newtheorem{theorem}{Theorem}
\newtheorem{definition}{Definition}

\newtheorem*{eg}{Example}

\renewcommand\hat\widehat
\renewcommand\geq\geqslant
\renewcommand\leq\leqslant
\renewcommand\bar\overline

\usepackage{tikz}
\usetikzlibrary{fadings}
\usetikzlibrary{patterns}
\usetikzlibrary{shadows.blur}
\usetikzlibrary{shapes}

\newcommand{\fig}[1]{{Fig.~\ref{#1}}}
\newcommand{\bpara}[1]		{\smallskip \noindent {\bf #1}}

\def\frconv					{*_\theta}

\def\Z					{\mathbb Z}

\def\iZ					{\in \mathbb Z}
\def\iR					{\in \mathbb R}

\def\DE					{\stackrel{\rm{def}}{=}}

\def\l						{\left(}
\def\r						{\right)}
\newcommand\rob[1]			{\l #1 \r}
\newcommand\sqb[1]		{\left[#1\right]}

\def\eg					{\emph{e.g.~}}

\DeclareMathOperator{\sinc}{sinc}

\def\awgn               	{\epsilon}
\def\cov                    {$cov$}

\DeclareFontFamily{U}  			{MnSymbolA}{}
\DeclareFontFamily{U}  			{MnSymbolC}{}
\DeclareSymbolFont{MnSyA}         	{U}  {MnSymbolA}{m}{n}
\DeclareSymbolFont{MnSyC}         	{U}  {MnSymbolC}{m}{n}

\DeclareFontShape{U}{MnSymbolA}{m}{n}{
    <-6>  MnSymbolA5
   <6-7>  MnSymbolA6
   <7-8>  MnSymbolA7
   <8-9>  MnSymbolA8
   <9-10> MnSymbolA9
  <10-12> MnSymbolA10
  <12->   MnSymbolA12}{}
  
  \DeclareFontShape{U}{MnSymbolC}{m}{n}{
    <-6>  MnSymbolC5
   <6-7>  MnSymbolC6
   <7-8>  MnSymbolC7
   <8-9>  MnSymbolC8
   <9-10> MnSymbolC9
  <10-12> MnSymbolC10
  <12->   MnSymbolC12}{}

\DeclareMathSymbol{\hpup}{\mathrel}{MnSyA}{56}
\DeclareMathSymbol{\hpdn}{\mathrel}{MnSyA}{50}

\DeclareMathSymbol{\crc}{\mathrel}{MnSyC}{77}

\newcommand{\upo}[1]	{ {\overset{\lower0.5em\hbox{$\smash{\scriptstyle	\hpup}$}} 	{{#1}}}}
\newcommand{\dno}[1]	{ {\overset{\lower0.5em\hbox{$\smash{\scriptstyle  	\hpdn}$}} 	{{#1}}}}
\newcommand{\up}[2]	{ {\overset{\lower0.2em\hbox{$\smash{\scriptstyle	\hpup}$}} 	{{#1}}} \left( {#2} \right)}
\newcommand{\dn}[2]	{ {\overset{\lower0.2em\hbox{$\smash{\scriptstyle	\hpdn}$}} 	{{#1}}} \left( {#2} \right)}
\newcommand{\upd}[2]	{ {\overset{\lower0.2em\hbox{$\smash{\scriptstyle	\hpup}$}} 	{{#1}}} \sqb{#2}}
\newcommand{\dnd}[2]	{ {\overset{\lower0.2em\hbox{$\smash{\scriptstyle	\hpdn}$}} 	{{#1}}} \sqb{#2}}

\newcommand{\upm}[2]	{\xi_\theta \rob{#2} {#1}\rob{#2}}
\newcommand{\dnm}[2]	{\xi^*_\theta \rob{#2} {#1}\rob{#2}}

\newcommand{\parder}[2]{\frac{\partial #2}{\partial #1}}
\newcommand{\cjg}[1]{{#1}^H}

\newcommand{\upx}[1]	{\xi_\theta \rob{#1}}
\newcommand{\dnx}[1]	{\xi^*_\theta \rob{#1}}

\newcommand{\ply}[1]	{\mathsf{Q}\rob{#1}}
\newcommand{\mat}[1]	{\mathbf{#1}}

\def\yt				{y_\crc}
\def\yq				{y_\crc^\mathsf{Q}}
\def\yqm				{\mat{y}_\crc^\mathsf{Q}}
\def\cm				{\breve{\mat{c}}^\theta}
\def\yn				{z}
\def\cov		    {\mathsf{cov}}
\def\covT		    {\mathsf{cov}\rob{\boldsymbol{\Theta}}}

\def\JT		       {\mathsf{J}\rob{\boldsymbol{\Theta}}}

\def\gra            {\Phi}
\def\snr            {\mathsf{PSNR}}
\def\var            {\mathsf{var}}

\newcommand{\JTK}[1]{\mathsf{J}_{#1}^{\boldsymbol{\Theta}}}

\newcommand{\diff}[2]{\Delta^{#1}#2}

\newcommand{\frft}[2]	{\mathscr{F}^{\left( #1 \right)}_{#2}}
\newcommand{\frfts}[1]	{\widehat{ #1 }_\theta \rob{\omega}}

\newcommand{\EQc}[1]		{\stackrel{(\ref{#1})}{=}}

\usepackage{xspace}
\def\frt					{FrFT\xspace}
\def\dfrt					{DFrFT\xspace}

\newcommand{\Sm}[1]		{\EuScript{S}_{#1}\rob{x}}
\newcommand{\Smd}[2]		{\EuScript{S}_{#1}^{\rob{#2}}\rob{x}}

\begin{document}

\title{Sparse Sampling in Fractional Fourier Domain: \\
Recovery Guarantees and Cram\'{e}r--Rao Bounds}

\author{V\'{a}clav Pavl\'{\i}\v{c}ek and Ayush Bhandari

\thanks{A.~Bhandari's work is supported by the UK Research and Innovation council's \emph{Future Leaders Fellowship} program ``Sensing Beyond Barriers'' (MRC Fellowship award no.~MR/S034897/1).}

\thanks{The authors are with the Dept. of Electrical and Electronic Engineering, Imperial College London, South Kensington, London SW7 2AZ, UK. (Emails: \texttt{vaclav.pavlicek20@imperial.ac.uk} and \texttt{ayush@alum.mit.edu}).}
\thanks{Manuscript submitted on January 2, 2024; accepted April 20, 2024. Date of the current version April 29, 2024. }
}

\markboth{Accepted with minor revisions.}
{Shell \MakeLowercase{\textit{et al.}}: Bare Demo of IEEEtran.cls for IEEE Journals}
\maketitle

\begin{abstract}
Sampling theory in fractional Fourier Transform (\frt) domain has been studied extensively in the last decades. 
This interest stems from the ability of the \frt to generalize the traditional Fourier Transform, broadening the traditional concept of bandwidth and accommodating a wider range of functions that may not be bandlimited in the Fourier sense.
Beyond bandlimited functions, sampling and recovery of sparse signals has also been studied in the \frt domain.
Existing methods for sparse recovery typically operate in the transform domain, capitalizing on the spectral features of spikes in the \frt domain.
Our paper contributes two new theoretical advancements in this area. 
First, we introduce a novel time-domain sparse recovery method that avoids the typical bottlenecks of transform domain methods, such as spectral leakage. This method is backed by a sparse sampling theorem applicable to arbitrary \frt-bandlimited kernels and is validated through a hardware experiment.
Second, we present Cram\'{e}r--Rao Bounds for the sparse sampling problem, addressing a gap in existing literature.
\end{abstract}


\IEEEpeerreviewmaketitle
\bigskip
\bigskip

\bigskip
\bigskip
\tableofcontents

\newpage
\linespread{1.2}

\section{Introduction}

\IEEEPARstart{T}{he} early 20th century saw the emergence of the fractional Fourier transform (\frt) as a significant development in harmonic analysis, pioneered by the works of Wiener \cite{Wiener:1929:J} and Condon \cite{Condon:1937:J}. The \frt generalizes the traditional Fourier transform (FT) by introducing the concept of \emph{fractionalization} explained next. Consider a time-domain function $f(t)$ and its FT, $\mathscr{F}_f (\omega) = \widehat{f}(\omega)$. Let 
$\frft{n}{f} = \frft{n-1}{f}\circ \mathscr{F}_f,n\geq1$ 
define the recursive application of the FT operator, where $\circ$ denotes operator composition and $\frft{0}{f} = \mathrm{I}$ (identity). It can be verified via basic Fourier analysis that,
\[\arraycolsep=8pt
\begin{array}{*{20}{c}}
\underbrace{ \frft{1}{f} = \widehat{f} }_{n=1}
&
\underbrace{\frft{2}{f} = \widetilde{f}}_{n=2}
&
\underbrace{ \frft{3}{f} = \widetilde{\widehat{f}}}_{n=3}
&
\underbrace{\frft{4}{f} = \frft{0}{f}}_{n=4}
\end{array},\]
where $\widetilde{f}(x) = f(-x)$. Iterating the Fourier operator shows that $\{\frft{n}{f}\}_{n\geq1}$ is 4-periodic automorphism or the FT is \emph{cyclic on a group of four} \cite{Condon:1937:J}. This periodicity of the Fourier operator is visually depicted in \fig{fig:frft}.

\begin{figure}[!t]
\centering
     \scalebox{0.9}{\tikzset{every picture/.style={line width=0.75pt}} 

\begin{tikzpicture}[x=0.75pt,y=0.75pt,yscale=-1,xscale=1]

\draw  [color={rgb, 255:red, 0; green, 0; blue, 255 }  ,draw opacity=1 ][fill={rgb, 255:red, 155; green, 155; blue, 155 }  ,fill opacity=0.05 ][dash pattern={on 4.5pt off 4.5pt}][line width=0.75]  (167,159.5) .. controls (167,116.7) and (201.7,82) .. (244.5,82) .. controls (287.3,82) and (322,116.7) .. (322,159.5) .. controls (322,202.3) and (287.3,237) .. (244.5,237) .. controls (201.7,237) and (167,202.3) .. (167,159.5) -- cycle ;
\draw [color={rgb, 255:red, 255; green, 0; blue, 0 }  ,draw opacity=1 ][line width=1.5]    (244.5,158) -- (312,158)(244.5,161) -- (312,161) ;
\draw [shift={(322,159.5)}, rotate = 180] [fill={rgb, 255:red, 255; green, 0; blue, 0 }  ,fill opacity=1 ][line width=0.08]  [draw opacity=0] (11.61,-5.58) -- (0,0) -- (11.61,5.58) -- cycle    ;
\draw [color={rgb, 255:red, 255; green, 0; blue, 0 }  ,draw opacity=1 ][line width=1.5]    (244.5,161) -- (177,161)(244.5,158) -- (177,158) ;
\draw [shift={(167,159.5)}, rotate = 360] [fill={rgb, 255:red, 255; green, 0; blue, 0 }  ,fill opacity=1 ][line width=0.08]  [draw opacity=0] (11.61,-5.58) -- (0,0) -- (11.61,5.58) -- cycle    ;
\draw [color={rgb, 255:red, 255; green, 0; blue, 0 }  ,draw opacity=1 ][line width=1.5]    (246,159.5) -- (246,227)(243,159.5) -- (243,227) ;
\draw [shift={(244.5,237)}, rotate = 270] [fill={rgb, 255:red, 255; green, 0; blue, 0 }  ,fill opacity=1 ][line width=0.08]  [draw opacity=0] (11.61,-5.58) -- (0,0) -- (11.61,5.58) -- cycle    ;
\draw [color={rgb, 255:red, 255; green, 0; blue, 0 }  ,draw opacity=1 ][line width=1.5]    (243,159.5) -- (243,92)(246,159.5) -- (246,92) ;
\draw [shift={(244.5,82)}, rotate = 90] [fill={rgb, 255:red, 255; green, 0; blue, 0 }  ,fill opacity=1 ][line width=0.08]  [draw opacity=0] (11.61,-5.58) -- (0,0) -- (11.61,5.58) -- cycle    ;
\draw [color={rgb, 255:red, 0; green, 0; blue, 255 }  ,draw opacity=1 ][line width=1.5]    (244.5,159.5) -- (329,84) ;
\draw [shift={(244.5,159.5)}, rotate = 318.22] [color={rgb, 255:red, 0; green, 0; blue, 255 }  ,draw opacity=1 ][fill={rgb, 255:red, 0; green, 0; blue, 255 }  ,fill opacity=1 ][line width=1.5]      (0, 0) circle [x radius= 4.36, y radius= 4.36]   ;
\draw  [draw opacity=0][line width=1.5]  (279.47,127.06) .. controls (284.76,129.25) and (289.42,132.96) .. (292.73,138.06) .. controls (296.47,143.79) and (297.91,150.37) .. (297.3,156.79) -- (266.55,155.1) -- cycle ; \draw [color={rgb, 255:red, 0; green, 0; blue, 255 }  ,draw opacity=1 ][line width=1.5]    (283.09,128.85) .. controls (286.88,131.07) and (290.2,134.16) .. (292.73,138.06) .. controls (296.47,143.79) and (297.91,150.37) .. (297.3,156.79) ;  \draw [shift={(279.47,127.06)}, rotate = 34.24] [fill={rgb, 255:red, 0; green, 0; blue, 255 }  ,fill opacity=1 ][line width=0.08]  [draw opacity=0] (8.13,-3.9) -- (0,0) -- (8.13,3.9) -- cycle    ;

\draw (276.23,146.2) node  [font=\large,color={rgb, 255:red, 0; green, 0; blue, 255 }  ,opacity=1 ]  {$\theta $};
\draw (370,72) node  [font=\large,color={rgb, 255:red, 0; green, 0; blue, 255 }  ,opacity=1 ]  {$\underbrace{\mathscr{F}_{f}^{( \theta )} =\hat{f}_{\theta }( \omega )}_{\mathsf{FrFT\ Domain}}$};
\draw (375.5,160) node  [font=\large]  {$\mathscr{F}_{f}^{( 0)} =f( t)$};
\draw (106.5,160) node  [font=\large]  {$\mathscr{F}_{f}^{( 2)} =f( -t)$};
\draw (246.5,257) node  [font=\large]  {$\mathscr{F}_{f}^{( 3)} =\hat{f}( -\omega )$};
\draw (241.5,66) node  [font=\large]  {$\mathscr{F}_{f}^{( 1)} =\hat{f}( \omega )$};
\draw (395,184) node  [font=\large]  {$=\mathscr{F}_{f}^{( 4)}$};
\draw (275,173.45) node  [font=\small] [align=left] {\textsf{Time}};
\draw (230.5,124) node  [font=\small,rotate=-270] [align=left] {\textsf{Frequency}};

\end{tikzpicture}}  

\caption{Visual illustration of ``fractionalizing'' the Fourier transform (FT) which introduces the \frt for non-integer orders, $\theta\in\mathbb{R}$. While the traditional FT is a 4-periodic automorphism \cite{Condon:1937:J} (applying FT 4 times returns the original function), this property holds only for integer orders, $n\in\mathbb{Z}$. The \frt generalizes this concept to arbitrary real orders $\theta$. When $\theta=\pi/2$, the \frt simplifies back to the FT.}
\label{fig:frft}
 \end{figure}
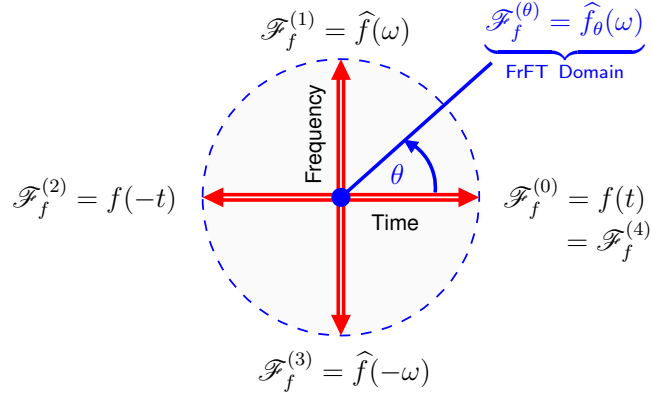

Given this cyclical nature of the FT operator, an intriguing question then is how to define a Fourier-like mapping on a continuous circle. Condon's work \cite{Condon:1937:J} provides the answer, introducing an integral transformation that applies to any arbitrary point $\theta\in\mathbb{R}$ on the circle. This leads to the mathematical definition of the \frt which is defined as a mapping, 
\begin{equation}
\label{eq:frft}
\frft{\theta}{f}: f \to \widehat{f}_\theta, 
\qquad
\frfts{f} = \int {f \rob{t} \kappa_\theta^\ast\left( {t,\omega } \right)dt},
\end{equation}
where the \frt kernel parameterized by $\theta\iR$ is written as, 
\begin{equation}
\label{kernel}
{\kappa _\theta }\left( {t,\omega } \right) \DE 
\begin{cases}
K_\theta{e^{ - \jmath \left( {\frac{{{t^2} + {\omega ^2}}}{2}\cot \theta  - \frac{{\omega t}}{{\sin \theta }}} \right)}}
& \theta \not = n\pi \\
\delta \left( {t - \omega } \right) & \theta = 2n\pi \\
\delta \left( {t + \omega } \right) & \theta+\pi = 2n\pi, \\
\end{cases}
\end{equation}
and $K_\theta$ is a known constant. The inverse-FrFT is the \frt with $-\theta$, which can be interpreted as rotation in $t$--$\omega$ plane (see \fig{fig:frft}). Clearly, at $\theta = \pi/2$, \frt simplifies to the FT.

Since its inception, the \frt has been widely applied in various fields. 
For comprehensive coverage, the reference book \cite{Ozaktas:2001:Book} by Ozaktas \& co-workers is an essential resource.
Notably, the \frt offers advantages in digital communications \cite{Martone:2001:J} and radar systems \cite{Yetik:2003:J}, where chirp-like basis functions (see \eqref{kernel}) are more effective than traditional sinusoidal basis (Fourier). Other areas of application include quantum mechanics  \cite{Namias:1961:J}, harmonic analysis \cite{Zayed:2002:J, Zayed:1996:J}, time-frequency representations \cite{Almeida:1994:J,Pei:2007:J}, and optics and imaging \cite{Amein:2007:C}. 

In the signal processing and harmonic analysis fields, a fundamental interest lies in understanding the connection between continuous and discrete representations through sampling theory. In the context of the \frt, key focus has been on the following two fronts.

\renewcommand\descriptionlabel[1]{\hspace{-0pt}{\raisebox{-2pt}\FilledDiamondShadowC} \hspace{0.5pt}\bf{#1}}

\begin{description}[leftmargin = 0pt,itemsep = 4pt,labelwidth=!,labelsep = !]
\item{\bf Shannon's Sampling.} This involves understanding how to represent bandlimited functions in the \frt domain as point-wise samples, considering that these functions are typically non-bandlimited in the FT domain. Several proofs of the \frt sampling theorem have been presented \cite{Xia:1996:J,Torres:2006:J,Tao:2008:Ja,Candan:2003:J,Garcia:2000:J,Wei:2010:J,Erseghe:1999:J,Bhandari:2012:J,Bhandari:2019:J}. Additionally, sampling theory has expanded beyond bandlimited spaces to include sparse signals \cite{Bhandari:2010:J,Bhandari:2015:C,Bhandari:2019:Ja} and shift-invariant spaces \cite{Bhandari:2012:J,Bhandari:2019:J}, and more recently to the Unlimited Sensing Framework \cite{Zhang:2023:C,Bhandari:2020:J,Bhandari:2017:C}. The main result in the area is as follows. 

\begin{theorem}[Sampling Theorem for \frt Domain \cite{Xia:1996:J,Torres:2006:J,Tao:2008:Ja,Candan:2003:J,Garcia:2000:J,Wei:2010:J,Erseghe:1999:J,Bhandari:2012:J,Bhandari:2019:J}]
\label{thm:frt}
If the function $f(t)$ contains no frequencies higher than $\Omega$ in the fractional Fourier domain, then $f(t)$ is completely determined by its equidistant samples spaced $T \leq \rob{\pi \sin \theta} /\Omega$. 
\end{theorem} 

Note that, at $\theta = \pi/2$, Theorem~\ref{thm:frt} reduces to the Shannon's sampling theorem for the Fourier domain. Recovery of the continuous-time function from samples $\{ f\rob{nT} \}_{n\in\Z}$ is facilitated by the \frt interpolation formula, 
\begin{equation*}
\label{eq:frftrec}
f\left( t \right) = {e^{ - \jmath \frac{{{t^2}}}{2}\cot \theta }}\sum\limits_{{n} \in \mathbb{Z}} {f\left( {{n}T} \right){e^{\jmath \tfrac{{{{\left( {{n}T} \right)}^2}}}{2}\cot \theta }}} \operatorname{sinc} \left( {\frac{t}{T} - {n}} \right).
\end{equation*}

\item{\bf Discretization of Continuous Transforms via Sampling.} Here, the focus has been to achieve discrete representations for the continuous transform. Several papers (see \cite{Pei:1999:Ja,Candan:2000:J,Liu:2014:J}, and references therein) have worked towards a definition of the discrete fractional Fourier transform (\dfrt), extending the discrete Fourier transform (DFT) in a manner analogous to how the continuous \frt expands upon the continuous FT. 
\end{description}
 
\bpara{Motivation and Contributions.} 
Sampling theory in the \frt domain is a prominent research area, as evidenced by the density of papers on this theme \cite{Xia:1996:J,Torres:2006:J,Tao:2008:Ja,Candan:2003:J,Garcia:2000:J,Wei:2010:J,Erseghe:1999:J,Bhandari:2012:J,Bhandari:2019:J}. Within this framework, the sampling and reconstruction of \emph{sparse signals} \cite{Bhandari:2010:J,Bhandari:2015:C,Bhandari:2019:Ja} has been studied in a variety of contexts due to the pervasive presence of such signals \cite{DeFigueiredo:1982:J,Vetterli:2002:J} in diverse application areas. The generalization of sparse recovery methods beyond Fourier domain \cite{Bhandari:2019:Ja,Liu:2014:J} is an interesting study in its own right. 
Current methods \cite{Bhandari:2010:J,Bhandari:2015:C,Bhandari:2019:Ja} suffer from two research gaps.
\begin{enumerate}[label  = $\bullet$ ]
  \item Sparse recovery methods rely on transform domain representations, where an unknown sparse signal corresponds to a sum-of-complex exponentials in \dfrt representation \cite{Bhandari:2010:J,Bhandari:2015:C,Bhandari:2019:Ja}, with recovery depending on spectral estimation methods  \cite{DeFigueiredo:1982:J,Vetterli:2002:J}. However, this approach is limited due to spectral leakage in discrete transforms like \dfrt \cite{Pei:1999:Ja,Candan:2000:J,Liu:2014:J} and DFT, especially in sub-Nyquist sampling scenarios. A  time-domain method would not suffer with such limitations. 
  \item Current works \cite{Bhandari:2010:J,Bhandari:2015:C,Bhandari:2019:Ja} are focused on sampling criterion with ideal measurements. Study of performance gurantees, a key tool for benchmarking algorithmic performance in the presence of noise, has not been discussed in the literature.
\end{enumerate}

This paper covers the theoretical aspects of the above research gaps. Our main contributions are as follows: 
\begin{enumerate}[ label  = $\bullet$ ]
\item We present a novel, time-domain sparse sampling approach that does not require any \dfrt operations and hence, its side-effects \eg spectral leakage, are eliminated. Our time-domain recovery approach is backed by a sampling theorem. 
\item We derive Cram\'{e}r--Rao Bounds (CRB) for the sparse parameter estimation problem. This serves as a performance guarantee for the recovery problem in the presence of noise. 
\end{enumerate}

\section{Mathematical Preliminaries}
Leading up to our main results, we will first revisit the basic mathematical operations associated with the FrFT domain.

\begin{definition}[Chirp Modulation] For a given \frt order $\theta\iR$, we define the quadratic phase, chirp modulation function by ${\xi _\theta }\left( t \right) = {\exp\rob{\jmath \frac{{\cot \theta }}{2}{t^2}}}$. For any function $f\rob{t}$ we define up-chirp and down-chirp operations, respectively, by 
\begin{equation}
\label{eq:chpmod}
\up{f}{t} \DE \upm{f}{t} \quad \mbox{ and } \quad \dn{f}{t} \DE \dnm{f}{t}.
\end{equation}
\end{definition}

When working with low-pass filters, we will be using the fractional convolution operator introduced by Zayed \cite{Zayed:1998:J}.

\begin{definition}[Fractional Convolution Operator \cite{Zayed:1998:J}] 
Given functions $f$ and $g$, let $*$ denote the conventional convolution operator, that is 
$\left( {f*g} \right)\left( t \right) = \int {f\left( x \right)g\left( {t - x} \right)dx}$. Then, the \frt convolution operator denoted by $\frconv$, is defined as
\begin{equation}
\label{eq:frc}
h\rob{t} = \rob{f \frconv g}\rob{t} \DE C_\theta \dnx{t} \rob{\upo{f}*\upo{g}}\rob{t},
\end{equation}
where ${C_\theta } = \sqrt {\left( {1 - \jmath \cot \theta } \right)/\left( {2\pi } \right)}$.
\end{definition}
The $\frconv$ operator \eqref{eq:frc} admits convolution-multiplication duality. Given  $h = f \frconv g$ in the time domain yields $\frfts{h} = \dnx{\omega} \frfts{f}\frfts{g}$ in \frt domain  \cite{Zayed:1998:J}. In our work, for simplicity of exposition, we will omit $C_\theta$ in our computations. 

\bpara{Signal Model and Measurements.} The unknown, continuous-time $K$-sparse signal is defined as,
\begin{equation}
\label{eq:ksp}
    s_K(t) = \sum\limits_{k=0}^{K-1} c_k \delta \left ( t - t_k \right ), \quad \sum\nolimits_k |c_k| <\infty.
\end{equation}
We will follow the measurement setup in \cite{Bhandari:2010:J,Bhandari:2015:C,Bhandari:2019:Ja} where one observes low-pass filtered version of the spikes in \eqref{eq:ksp} in the \frt sense. This is consistent with the conventional sparse sampling and super-resolution problems \cite{DeFigueiredo:1982:J,Vetterli:2002:J} in the Fourier domain. Concretely, we will assume that the sampling kernel is $\phi_M\rob{t} = \dno{\psi}_M(t)$ is bandlimited in the \frt domain. We will be working with \emph{arbitrary kernels} of the form, $\psi_M (t) = \sum\nolimits_{m = 0}^{M - 1} {{p_m}\operatorname{sinc} \left( {\frac{t}{T} - m} \right)} $ where the weights $\{p_m\}_{m=0}^{M-1}$ are assumed to be known. The low-pass filtered measurements simplify to, 
\begin{align}
\label{eq:yt}
y\rob{t} &= \rob{s_K \frconv \phi_M}\rob{t}  \\
&\EQc{eq:frc} \dnx{t} 
\sum\limits_{k=0}^{K-1} c_k
\int_{-\infty}^{\infty} \upx{\tau}  {\delta(\tau - t_k) \phi_M(t - \tau)} d\tau \notag \\
& = \dnx{t} \sum\limits_{k=0}^{K-1} c_k \upx{t_k} \sum\limits_{m=0}^{M-1} p_m \sinc \left ( \frac{t-t_k}{T} - m\right ). \notag
\end{align}

\bpara{Problem Statement.}
Given $N$ samples, $y\sqb{n} = y\rob{nT}, T>0$, our goal is to estimate the unknown sparse signal $s_K\rob{t}$.

\section{Recovery Algorithm and Guarantees}
Our sparse recovery result is as follows.
\begin{theorem}[Sampling Criterion]
\label{thm:samp-crit}
Let $y\rob{t} = \rob{s_K \frconv \phi_M}\rob{t}$, as in \eqref{eq:yt} and $\phi_M,\theta$ be known. Suppose we are given $N$ samples of $y\sqb{n} = y\rob{nT}, T>0$, then, $N\geq 2KM$ guarantees recovery of the unknown signal $s_K\rob{t}$ in \eqref{eq:ksp}.
\end{theorem}

\begin{proof}
We provide a proof by construction. Let $\bar{t}_k \DE t_k/T$. Given samples $y\sqb{n} = y\rob{nT}, T>0$, by modulating both sides with $\upx{nT}$ we obtain,
\begin{equation}
\label{eq:smodel}
\underbrace{y\sqb{n} {e^{\jmath \frac{{\cot \theta }}{2}{\rob{nT}^2}}}}_{\upd{y}{n}}
 \EQc{eq:yt} \sum\limits_{k=0}^{K-1} c_k \upx{t_k} \sum\limits_{m=0}^{M-1} p_m \sinc \rob{n - \bar{t}_k - m}.
\end{equation} 
Since $\forall n \in \mathbb{Z}$, $\sin \left( {\pi \left( {n - x} \right)} \right) = {\left( { - 1} \right)^{n + 1}}\sin \left( {\pi x} \right)$ and $\sinc \left ( n-\bar{t}_k - m \right) = \frac{(-1)^{n-m+1} \sin \left (\pi \bar{t}_k \right )}{\pi \left (n - m - \bar{t}_k\right )}$, we obtain via \eqref{eq:smodel}, 
\[
\upd{y}{n} \EQc{eq:yt} 
 (-1)^{n+1}
\sum\limits_{k=0}^{K-1}
c_k \upx{t_k} 
\sum\limits_{m=0}^{M-1} p_m \frac{(-1)^{-m} \sin(\pi \bar{t}_k)}{\pi (n - m - \bar{t}_k)}.
\]
A re-arrangement of the terms leads to, 
\begin{equation}
\label{eq:z_n_definition}
\yt\sqb{n} \DE
\rob{\frac
{\pi y\sqb{n} {e^{\jmath \frac{{\cot \theta }}{2}{\rob{nT}^2}}}}
{(-1)^{n+1}} }
\EQc{eq:smodel}
\sum\limits_{k=0}^{K-1} \breve{c}^\theta_k \sum\limits_{m=0}^{M-1}  \frac{(-1)^{-m}p_m}{\rob{n - m - \bar{t}_k}},	
\end{equation}
where $\breve{c}^\theta_k \DE c_k \sin \left (\pi \bar{t}_k \right ){\exp\rob{\jmath \frac{{\cot \theta }}{2}{t_k^2}}}$. Let $\ply{z}$ be a polynomial of order $KM$ with its roots $\{\bar{t}_k\}$, of multiplicity $M$,
\begin{equation}
\label{eq:polyKM}
\ply{z} \DE 
\sum\limits_{k=0}^{\rob{KM}} q\sqb{k} z^k = \prod\limits_{k=0}^{K-1}\prod\limits_{m=0}^{M-1}\left ( z - m - \bar{t}_k\right ).
\end{equation}
With $p_m^\prime \DE (-1)^{-m}p_m$, we may now write, 
\begin{equation}
\label{eq:ann1}
\underbrace{\yt\sqb{n} \ply{n}}_{\yq\sqb{n}}
= \sum\limits_{k=0}^{K-1} \breve{c}^\theta_k \underbrace{ 
\sum\limits_{m=0}^{M-1}p_m^\prime
\frac{\prod\limits_{k=0}^{K-1}\prod\limits_{m=0}^{M-1}\left ( n - m - \bar{t}_k\right )}
{n - m - \bar{t}_k}
}_{P_{n,k}},
\end{equation}
which can be written as $\yqm = \mat{P}\cm$ in vector-matrix notation. Note that, $\yq\sqb{n}=\sum\nolimits_{k=0}^{K-1} \breve{c}^\theta_k P_{n,k}$ is a polynomial of order $\rob{KM-1}$ because $\left ( n - m - \bar{t}_k\right )$ in the denominator of \eqref{eq:ann1} peels away a monomial. Let $\rob{\Delta f}[n] = f[n+1] - f[n]$ whose recursive application results in $\Delta^L = \Delta^{L-1} \Delta$.
The implication is that 
$\Delta^{KM}\rob{\sum\nolimits_{k=0}^{K-1} \breve{c}^\theta_k P_{n,k}}\sqb{n} = 0$. Consequently, 
\begin{equation*}
    \label{eq:ann3}
    \rob{\diff{KM}\yq}\sqb{n} = \sum\limits_{k=0}^{KM} {q\sqb{k}}
    \underbrace{\diff{KM}{\left ( n^k \yt\sqb{n}  \right )}}_{D_{n,k}} = 0 \equiv \mat{Dq} = \mat{0},
\end{equation*}
where $\mat{D}$ is the matrix purely depending on data samples $y[n]$,
\[
\sqb{\mat{D}}_{n,k} = 
\diff{KM}{\left ( n^k \yt\sqb{n}  \right )} 
\EQc{eq:z_n_definition} 
\diff{KM}
\rob{n^k
\frac
{\pi{e^{\jmath \frac{{\cot \theta }}{2}{\rob{nT}^2}}}}
{(-1)^{n+1}} 
y\sqb{n} },
\]
and known \frt order $\theta$. 
Clearly, given $\mat{y} \in \mathbb{C}^N$, $\diff{KM}$ leads to a reduction of $KM$ samples and $\mat{D} \in \mathbb{C}^{\rob{N-KM}\times \rob{KM+1}}$.
Let $\ker{\rob{\mat{M}}} = \{ \mat{z} \in \mathbb{C}^N \ | \ \mat{Mz} = \mat{0} \}$ and $\sqb{\mat{q}}_k \EQc{eq:polyKM} q\sqb{k}$, $\mat{q}\in\mathbb{R}^{KM+1}$. 
For $t_k\not=t_{\ell}, \forall k \not = \ell$, a non-zero $\mat{q} \in \ker{\rob{\mat{D}}}$ can be obtained provided that $\mat{D}$ is rank-deficient or whenever $N-KM \geq KM\Rightarrow N\geq 2KM$, which gives the sampling criterion. The solution to $\mat{Dq}= \mat{0}\mapsto \mat{q}$ allows for construction of $\ply{z}$ in \eqref{eq:polyKM}; its roots then lead to the estimates of $\{t_k\}_{k=0}^{K-1}$. It remains to estimate $\{c_k\}_{k=0}^{K-1}$ which can be obtained by solving linear system of equations in in \eqref{eq:smodel}, thus leading to exact recovery of unknown $s_K\rob{t}$ in \eqref{eq:yt}.
\end{proof}

\begin{figure}[!t]
\centering
\begin{overpic}[width=0.5\linewidth]{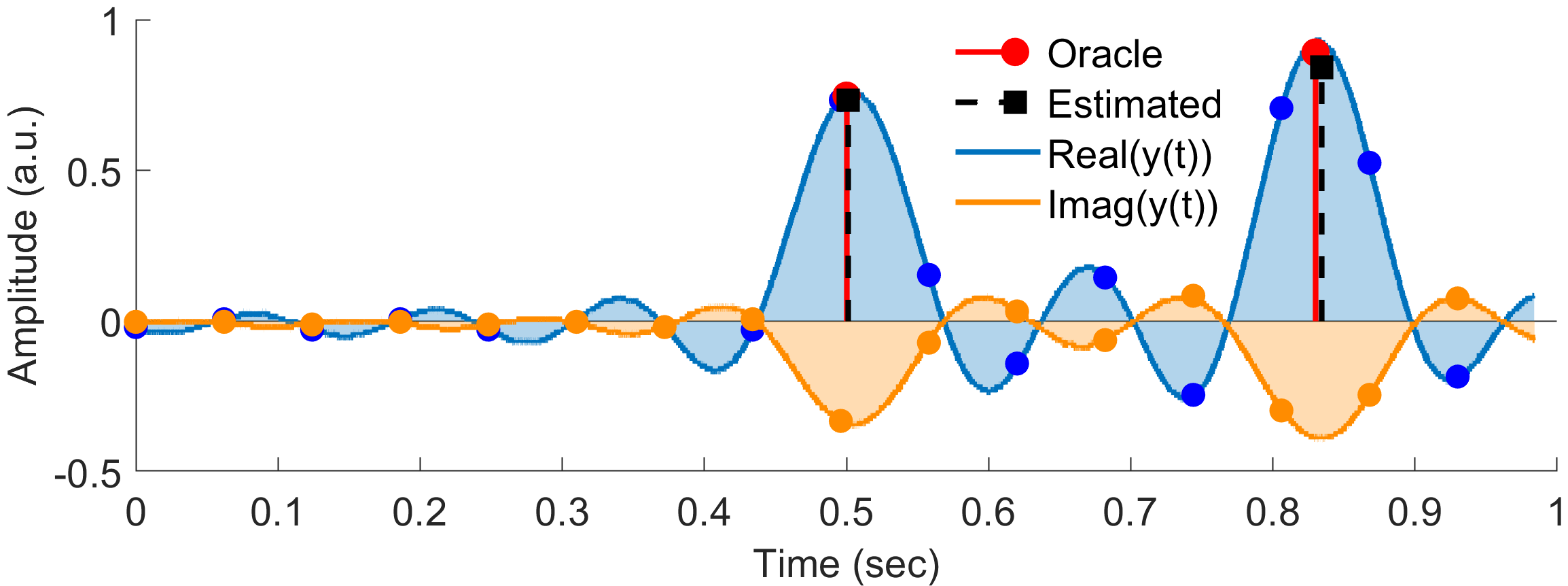}
    \put(95,335){\parbox{0.3in}{\centering \footnotesize $c_0$}}
    \put(95,300){\parbox{0.3in}{\centering \scriptsize 0.748}}
    \put(190,330){\parbox{0.3in}{\centering \footnotesize $c_1$}}
    \put(190,300){\parbox{0.3in}{\centering \scriptsize 0.891}}
    \put(285,330){\parbox{0.3in}{\centering \footnotesize $t_0 \;\text{(s)}$}}
    \put(285,300){\parbox{0.3in}{\centering \scriptsize 0.500}}
    \put(380,330){\parbox{0.3in}{\centering \footnotesize $t_1 \;\text{(s)}$}}
    \put(380,300){\parbox{0.3in}{\centering \scriptsize 0.830}}
    \put(95,250){\parbox{0.3in}{\centering \footnotesize $\tilde{c}_0$}}
    \put(95,220){\parbox{0.3in}{\centering \scriptsize 0.733}}
    \put(190,250){\parbox{0.3in}{\centering \footnotesize $\tilde{c}_1$}}
    \put(190,220){\parbox{0.3in}{\centering \scriptsize 0.843}}
    \put(285,250){\parbox{0.3in}{\centering \footnotesize $\tilde{t}_0 \;\text{(s)}$}}
    \put(285,220){\parbox{0.3in}{\centering \scriptsize 0.501}}
    \put(380,250){\parbox{0.3in}{\centering \footnotesize $\tilde{t}_1 \;\text{(s)}$}}
    \put(380,220){\parbox{0.3in}{\centering \scriptsize 0.834}}
\end{overpic}
\caption{Hardware-based experimental validation of sparse recovery method. Digital samples with $8$-bit resolution are marked by $\bullet$.}
\label{fig:hw-data}
 \end{figure}
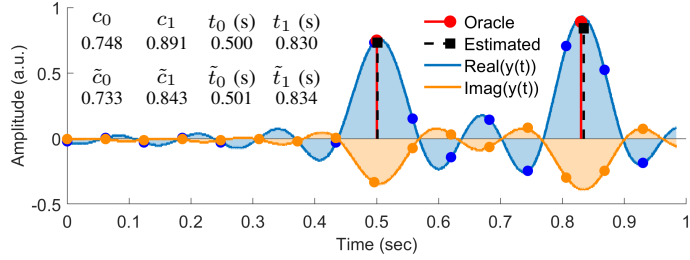

Note that when $\theta = \pi/2$, Theorem~\ref{thm:samp-crit} reduces to the conventional Fourier domain case and $\ast_\theta, \theta = \pi/2$ reduces to the conventional notion of low-pass filtering or convolution. Furthermore, when $M = 1$, our results coincide with the Fourier domain recovery result in \cite{Vetterli:2002:J}.

\bpara{Hardware Experiment.} Although our algorithm is not designed to tackle noise, we validate its empirical robusteness via hardware experiment. To this end, we generate $y(t)$ in \eqref{eq:yt} with \frt order $\theta = \pi/4$, $\psi_M(t)$ with $M = 1$ and $T = 0.062$ s and $s_K(t)$ with $t_k = \{0.50, 0.83\}$ and $c_k =  \{0.748, 0.891\}$ using \texttt{PicoScope 2204A}. 
To capture the complex-valued signal, we utilize in-phase and quadrature channels, digitizing each channel using an \mbox{8-bit} ADC.
We use $N = 16$ samples annotated by ``$\bullet$'' (blue and orange) in \fig{fig:hw-data}.
In this setting, the measurements are naturally affected by thermal and quantization noise, among other factors. 
Despite this, our algorithm estimates the sparse signal with mean squared error (MSE) for $c_k \propto 10^{-3}$ and $t_k\propto 10^{-6}$, respectively.

\section{Cram\'{e}r--Rao Bounds for Sparse Sampling}
To benchmark the recovery performance in the presence of noise, we derive CRB for the sparse parameter estimation problem. To this end, we define,
\begin{equation}
\yn\sqb{n} = y\sqb{n} + \awgn\sqb{n}, \quad 
\awgn \sim \mathsf{CN}\rob{0,\sigma^2},
\end{equation}
where $\mathsf{CN}$ denotes the white complex Gaussian noise distribution. The CRB is obtained by inverting the Fisher Information Matrix (FIM) parametrized by $\boldsymbol{\Theta}  = \sqb{t_0,\cdots,t_{K-1} \ | \ c_0,\cdots,c_{K-1}  }^{\top}$. The FIM is given by\cite{Porat:1986:J}
\begin{equation}
    \JT = \cjg{\mat{G}}\mat{R}^{-1}\mat{G} = \frac{1}{\sigma^2} \cjg{\mat{G}} \mat{G},
    \label{eq:covariance-matrix}
\end{equation}
where $\mat{G} = \parder{\boldsymbol{\Theta}}{\yn\sqb{n}}$ is a matrix of partial derivatives of $\yn\sqb{n}$ and $\mat{R}$ is the covariance matrix of the $\mathsf{CN}$. CRB matrix is obtained by inverting $\JT$, yielding $\cov \left ( \theta \right ) \geq \left ( \JT \right ) ^ {-1}$ where, 
\begin{align}
    \JT = \frac{1}{\sigma^2} \begin{bmatrix}
        \cjg{\left( \parder{t_k}{\yn\sqb{n}} \right )} \parder{t_m}{\yn\sqb{n}} & 
        \cjg{ \left( \parder{t_k}{\yn\sqb{n}} \right )} \parder{c_m}{\yn\sqb{n}} \\
        \cjg{\left( \parder{c_k}{\yn\sqb{n}} \right )}\parder{t_m}{\yn\sqb{n}} & \cjg{\left ( \parder{c_k}{\yn\sqb{n}} \right )} \parder{c_m}{\yn\sqb{n}} \\
    \end{bmatrix},
    \label{eq:fim-ders}
\end{align}
with $k, m \in \{ 0, 1, \hdots, K - 1\}$. With $\phi\rob{t} = \dno{\psi}(t)$, we can simplify the matrix elements in the above as follows:
\begin{align}
\frac{{\partial z}\sqb{n}}{{\partial {c_k}}}   & =\psi \left( {n - \frac{{{t_k}}}{T}} \right){e^{ - \jmath \frac{{\cot (\theta )}}{2}\left( {{{\left( {nT} \right)}^2} - t_k^2} \right)}} 
\label{eq:zct}\\
\frac{{\partial z}\sqb{n}}{{\partial {t_k}}}  & 
={c_k}\left( { - \frac{{{\psi ^{\left( 1 \right)}}\left( {n - {t_k}/T} \right)}}{{T\psi \left( {n - {t_k}/T} \right)}} + \jmath {t_0}\cot (\theta )} \right)
\rob{\frac{{\partial z \sqb{n}}}{{\partial {c_k}}}} \notag ,
\end{align}
where ${{\psi ^{\left( 1 \right)}}}\rob{t} \DE \partial_t \psi\rob{t}$. For the general case of arbitrary $K$, the analytical results are intractable. Hence, we derive the closed-form CRB for a single spike case.

\bpara{Analytical Bounds: Single Spike Case.} When $K = 1$, the FIM in \eqref{eq:fim-ders}
simplifies to, $\JT \EQc{eq:fim-ders} \frac{1}{\sigma^2} \begin{bsmallmatrix}
    \JTK{1} & \JTK{2} \\\JTK{3} & \JTK{4}
    \end{bsmallmatrix}$.
Next, with $\boldsymbol{\Theta} = \sqb{t_0 \ | \ c_0}^{\top}$, we investigate each matrix element separately. Using \eqref{eq:zct}, we simplify  the expressions for $\{\mathsf{J}_m^{\boldsymbol{\Theta}}\}_{m=1}^{m=4}$ read, 
\begin{align*}
 & \left\{{\begin{array}{*{20}{l}}
 \JTK{1} = \frac{{c_0^2}}{{{T^2}}}\left( {{T^2}{{\cot }^2}\left( \theta  \right)t_0^2\Phi _{\psi ,\psi }^N\left( \bar{t}_0 \right) + \Phi _{{\psi ^{\left( 1 \right)}},{\psi ^{\left( 1 \right)}}}^N\left( \bar{t}_0 \right)} \right) \\   
\JTK{2} = - \frac{c_0}{T} \left( \Phi_{\psi,\psi^{(1)}}^N\left( \bar{t}_0 \right) + \jmath T\cot \left( \theta  \right){t_0}\Phi_{\psi,\psi}^N \left ( \bar{t}_0 \right ) \right)\\
\JTK{3} = \frac{c_0}{T} \left(\jmath T\cot \left( \theta \right){t_0}\Phi_{\psi,\psi}^N \left ( \bar{t}_0 \right ) - \Phi_{\psi,\psi^{(1)}}^N \left ( \bar{t}_0 \right ) \right)  \\
 \JTK{4} = \Phi_{\psi,\psi}^N \left ( \bar{t}_0 \right ) 
\end{array}} \right.
 \mbox{ where } \
        \gra_{\psi_1,\psi_2}^N \left ( t \right ) = \sum\limits_{n = 0}^{N-1} \psi_1 \left ( n - t \right ) \psi_2 \left ( n - t \right ).
\label{J:gra} 
\tag{15}     
\setcounter{equation}{15}
\end{align*}
This leads to a method for the numerical computation for the CRB with arbitrary $\psi$. Note that $\{\mathsf{J}_m^{\boldsymbol{\Theta}}\}_{m=1}^{m=4}$ implicitly depend on $\gra_{\psi_1,\psi_2}^\infty$ in \eqref{J:gra}. With regards to Theorem~\ref{thm:samp-crit}, we can compute analytical bounds by setting $\psi = \sinc$ in the asymptotic case, $N\to \infty$. Hence, we compute the different parametrizations of $\gra_{\psi_1,\psi_2}^\infty$ to obtain the FIM in \eqref{eq:covariance-matrix}. To this end, we define, 
\begin{equation}
\label{eq:sm}
\Sm{m}  \DE \sum\limits_{n \in \mathbb{Z}} {\frac{1}{{{{\left( {x/T - n} \right)}^m}}}}  = \sum\limits_{n \in \mathbb{Z}} {\frac{1}{{{{\left( {\overline x  - n} \right)}^m}}}},
\end{equation}
which will be a common factor to all of our computations. 

\begin{description}[leftmargin = 0pt,itemsep = 10pt,labelwidth=!,labelsep = !]
\item{\bf Computation of ${\Phi _{\psi ,\psi }^\infty \left( \bar{t}_0 \right)} $.}  Let $x = \bar{t}_0 \in \mathbb{R}$ where $\bar{t}_0 = t_0/T$. Since $\forall n \iZ, \sin \left( {\pi \left( {\bar{x} - n} \right)} \right) = {\left( { - 1} \right)^n}\sin \left( {\pi \bar{x}} \right)$ is the constituent term in ${\Phi _{\psi ,\psi }^\infty \left( \bar{x} \right)}$, we can write,  
\[
{\left. {\Phi _{\psi ,\psi }^\infty \left( \bar{x} \right)} \right|_{x = {t_0}}} = \frac{{{{\sin }^2}\left( {\pi \bar{x}} \right)}}{{{\pi ^2}}}
\underbrace{\frac{1}{{\sum\nolimits_{n \in \mathbb{Z}} {{{\left( {\bar{x} - n} \right)}^2}} }}}_{\Sm{2}}
= 1\to \boxed{ \JTK{4}}
\]
because $\Sm{2}= {\pi ^2}{\csc ^2}\left( {\pi x} \right)$, which is a well-known identity in complex analysis (see \cite{Ahlfors:1979:B}, page 189).

\item{\bf Computation of ${\Phi _{\zeta, \zeta }^\infty \left( \bar{x} \right)}$ where $\zeta = \psi^{\rob{1}}$.} We begin by noting,  $\zeta \left( {\bar x} \right) = 
\pi \left( {\frac{{\cos \left( {\pi \bar x} \right)}}{{\pi \bar x}} - \frac{{\sin \left( {\pi \bar x} \right)}}{{{\pi ^2}{{\bar x}^2}}}} \right).$
The constituent term, 
\begin{align*}
\underbrace{\rob{{\zeta}\left( {\bar x - n} \right)}^2}_{\zeta =\psi^{\rob{1}}}
= &  \frac{{{{\sin }^2}\left( {\pi \left( {\bar x - n} \right)} \right)}}{{{\pi ^2}{{\left( {\bar x - n} \right)}^4}}} + \frac{{{{\cos }^2}\left( {\pi \left( {\bar x - n} \right)} \right)}}{{{{\left( {\bar x - n} \right)}^2}}} \hfill  - \frac{{2\sin \left( {\pi \left( {\bar x - n} \right)} \right)\cos \left( {\pi \left( {\bar x - n} \right)} \right)}}{{\pi {{\left( {\bar x - n} \right)}^3}}},
\end{align*}
can be simplified via basic trigonometric identities since, 
${\sin ^2}\left( {\pi \left( {\bar x - n} \right)} \right) = {\sin ^2}\left( {\pi \bar x} \right)$,
${\cos ^2}\left( {\pi \left( {\bar x - n} \right)} \right) = {\cos ^2}\left( {\pi \bar x} \right)$, and
$2\sin \left( {\pi \left( {\bar x - n} \right)} \right)\cos \left( {\pi \left( {\bar x - n} \right)} \right) =  \sin \left( {2\pi \bar x} \right)$, yielding,
\begin{equation}
\label{eq:pxixi}
\Phi _{\zeta,\zeta}^\infty \left( {\bar x} \right) = \sum\limits_{n \in \mathbb{Z}} {\left( {\frac{{{{\sin }^2}\left( {\pi \bar x} \right)}}{{{\pi ^2}{{\left( {\bar x - n} \right)}^4}}} + \frac{{{{\cos }^2}\left( {\pi \bar x} \right)}}{{{{\left( {\bar x - n} \right)}^2}}} - \frac{{\sin \left( {2\pi \bar x} \right)}}{{\pi {{\left( {\bar x - n} \right)}^3}}}} \right)}.
\end{equation}
In the above, note that the terms corresponding to $\{ \sin^2\rob{\pi\bar{x}}, \cos^2\rob{\pi\bar{x}}, \sin\rob{2\pi\bar{x}}\}$ are linked with the sums $\{ \Sm{4},\Sm{2}, \Sm{3} \}$ in \eqref{eq:sm}, respectively. Known that, $\Sm{2} = {\pi ^2}{\csc ^2}\left( {\pi \bar x} \right)$ (see \cite{Ahlfors:1979:B}, page 189). To compute $\Sm{3}$ and $\Sm{4}$, we proceed by differentiating $\Sm{2}$,
\begin{align*}
& \underbrace{{\partial _x \Sm{2}}}_{\Smd{2}{1}}  =  
- \frac{2}{T}
\underbrace{\sum\limits_{n \in \mathbb{Z}} {\frac{1}{{{{\left( {\bar x - n} \right)}^3}}}}}_{\Sm{3}} = 
\underbrace { - \frac{{2{\pi ^3}}}{T}\cot \left( {\pi \bar x} \right){{\csc }^2}\left( {\pi \bar x} \right)}_{{\partial _x}\left( {{\pi ^2}{{\csc }^2}\left( {\pi \overline x } \right)} \right)} \\
& \Rightarrow \Sm{3} = {\pi ^3}\cot \left( {\pi \bar x} \right){\csc ^2}\left( {\pi \bar x} \right).
\end{align*}
Similarly, 
\begin{align*}
& \underbrace{{\partial ^2 _x \Sm{2}}}_{\Smd{2}{2}}  =  
\rob{\frac{6}{T^2}}
{\Sm{4}} = 
\underbrace {
\frac{{2{\pi ^4}}}{{{T^2}}}\left( {\cos \left( {2\pi \overline x } \right) + 2} \right){\csc ^4}\left( {\pi \overline x } \right)
}_{{\partial^2 _x}\left( {{\pi ^2}{{\csc }^2}\left( {\pi \overline x } \right)} \right)} \\
& \Rightarrow \Sm{4} = 
\tfrac{\pi ^4}{3}\left( {\cos \left( {2\pi \overline x } \right) + 2} \right){\csc ^4}\left( {\pi \overline x } \right).
\end{align*}
Given, $\Phi _{\zeta,\zeta}^\infty \left( {\bar x} \right) = 
\rob{{\sin}\left( {\pi \bar x} \right)/\pi}^2\Sm{4} + {\cos ^2}\left( {\pi \bar x} \right)\Sm{2} - 
\rob{
\sin \left( {2\pi \bar x} \right)/\pi
}
\Sm{3}$ (see \eqref{eq:pxixi}), substituting for $\Sm{m}, m = \{2,3,4\}$ in \eqref{eq:pxixi} with $\bar{x} = x/T$, we obtain, 
\begin{equation}
\label{eq:pxiD}
    \boxed{\Phi _{{\psi ^{\left( 1 \right)}},{\psi ^{\left( 1 \right)}}}^\infty \left( {\bar x} \right)} = \sum\limits_{n \in \mathbb{Z}} \left ({{\psi ^{\left( 1 \right)}}\left( {\frac{x}{T} - n} \right)} \right )^2  = \frac{{{\pi ^2}}}{{3}}.
\end{equation}

\item{\bf Computation of ${\Phi _{\zeta ,\psi }^\infty \left( \bar{x} \right)}$ where $\zeta = \psi^{\rob{1}}$.} This term turns out to be zero. 
To see this, we expand the constituent term in $\Phi _{{\psi ^{\left( 1 \right)}},\psi }^\infty \left( {\bar x} \right)$, that is, $\psi \left( {\bar x - n} \right){\psi ^{\left( 1 \right)}}\left( {\bar x - n} \right) = \nicefrac{{\sin (\pi (\bar x - n))\cos (\pi (\bar x - n))}}{{\pi {{(\bar x - n)}^2}}} - \nicefrac{{{{\sin }^2}(\pi (\bar x - n))}}{{{\pi ^2}{{(\bar x - n)}^3}}}$. Using $\sin (\pi (\bar x - n))\cos (\pi (\bar x - n)) = \sin (2\pi \bar x) / 2$, we conclude,
\begin{align*}
 \boxed{ \Phi _{{\psi ^{\left( 1 \right)}},\psi }^\infty \left( {\bar x} \right)}
&    = \sum\limits_{n \in \mathbb{Z}} {\frac{{\sin (2\pi \bar x)}}{{2\pi {{(n - \bar x)}^2}}}}  - \frac{{{{\sin }^2}(\pi \bar x)}}{{{\pi ^2}{{(n - \bar x)}^3}}}  = \underbrace{\frac{{\sin (2\pi \bar x)}}{{2\pi }}\Sm{2}}_{\pi\cot\rob{\pi\bar{x}}} - 
\underbrace{\frac{{{{\sin }^2}(\pi \bar x)}}{{{\pi ^2}}}\Sm{3}}_{\pi\cot\rob{\pi\bar{x}}} = 0.
\end{align*} 
\end{description}
In summary, our computations show, 
\[\begin{array}{*{20}{c}}
{\Phi _{\psi ,\psi }^\infty \left( \bar{t}_0 \right)} = 1
&
{\Phi _{\psi^{(1)} , \psi^{(1)}}^\infty \left( \bar{t}_0 \right)} = \frac{\pi^2}{3}
&
\Phi _{{\psi ^{\left( 1 \right)}},\psi }^\infty \left( {\bar{t}_0 } \right) = 0.
\end{array}\]

\bpara{CRB Computation.} Since we have all the values of $\gra_{\psi_1,\psi_2}^\infty$ in \eqref{J:gra} that define $\{\mathsf{J}_m^{\boldsymbol{\Theta}}\}_{m=1}^{m=4}$, we can now obtain the FIM by back substituting the obtained values of $\gra_{\psi_1,\psi_2}^\infty$
yielding the FIM elements,
(i)~$\JTK{1}  = c_0^2\left( {{{\cot }^2}\left( \theta  \right)t_0^2 + \frac{{{\pi ^2}}}{{3{T^2}}}} \right)$,
(ii)~$\JTK{2}  = -\jmath c_0 \cot \left( \theta \right){t_0}$,
(iii)~$\JTK{3} = \jmath c_0 \cot \left( \theta \right){t_0}$, and,
(iv)~$\JTK{4} = 1$. The corresponding FIM takes the form of,
\begin{align*}
    \JT = \frac{1}{\sigma^2} \begin{bmatrix}
        c_0^2\left( {{{\cot }^2}\left( \theta  \right)t_0^2 + \tfrac{{{\pi ^2}}}{{3{T^2}}}} \right) & -\jmath c_0 \cot \left( \theta \right){t_0} \\
        \jmath c_0 \cot \left( \theta \right){t_0} & 1 \\
    \end{bmatrix},
\end{align*}
with ${\operatorname{det}\rob{\JT}}=
\left( {{\pi ^2}c_0^2} \right)/\left( {3{T^2}} \right)$. The CRB is given by,
\begin{equation}
    \covT = \frac{\sigma^2}{\operatorname{det}\rob{\JT}} \begin{bmatrix}
        \JTK{4} & -\JTK{2} \\ -\JTK{3} & \JTK{1}
    \end{bmatrix}
    \label{eq:cov-def}
\end{equation}
and using $\covT$ and defining $\snr =c_0^2/\sigma^2$, we obtain,
\begin{align*}
    \covT &= \rob{\frac{3{T^2}}{\pi^2 \snr} }\begin{bmatrix}
        1 & \jmath c_0 \cot \left( \theta \right){t_0} \\
        -\jmath c_0 \cot \left( \theta \right){t_0} & c_0^2\left( {{{\cot }^2}\left( \theta  \right)t_0^2 + \tfrac{{{\pi ^2}}}{{3{T^2}}}} \right)
    \end{bmatrix},
\end{align*}
giving,
\begin{empheq}[left = \empheqlbrace]{align*}
& \var\rob{t_0} \geq
\frac{3T^2}{\pi^2 \snr}\\[0.5em]
& \var\rob{c_0} \geq
\frac{3c_0^2T^2}{\pi^2 \snr} \rob { \rob{t_0 \cot\rob{\theta}}^2 + \frac{{{\pi ^2}}}{{3{T^2}}}}.
\end{empheq}

\section{Conclusion}
This paper introduces a new time-domain method for recovery of sparse signals from low-pass filtered measurements in the Fractional Fourier Transform (\frt) domain. The advantage of this method is that unlike previous methods that work in the \frt domain and are susceptible to spectral leakage, the new method overcomes restrictions. We also derive Cram\'{e}r--Rao Bounds (CRB) for the sparse estimation problem that was missing in previous literature. Future work in this area includes recovery with different classes of sampling kernels, particularly non-bandlimited ones and their performance evaluation in terms of Cram\'{e}r--Rao Bounds.

\ifCLASSOPTIONcaptionsoff
\newpage
\fi


\end{document}